\documentclass[11pt]{article}
\usepackage{amssymb}
\newcommand{\setmargins}[1]{ 
	\setlength{\textwidth}{8.5in} 
	\addtolength{\textwidth}{-#1}  \addtolength{\textwidth}{-#1}
	\setlength{\oddsidemargin}{-1in} \setlength{\evensidemargin}{-1in} 
	\addtolength{\oddsidemargin}{#1} \addtolength{\evensidemargin}{#1}
	}
\setmargins{1in}
\newcommand{\twosp}{} \newcommand{\onesp}{} 

\newenvironment{singlespace}{\onesp}{\twosp}
\newtheorem{thm}{Theorem}[section] 
\newtheorem{lemma}[thm]{Lemma}     
\newenvironment{proof}%
{\smallskip\par\noindent{\bf Proof:}\ }{$\Box$\medskip\par}
{\smallskip\par\noindent{\bf Proof of #1:}\ }{$\Box$\medskip\par}
\newcounter{qnum}

\begin{document}
\bibliographystyle{plain}

\title{A Randomized Sublinear Time Parallel GCD Algorithm for the EREW PRAM}

\author{
  Jonathan P.\ Sorenson\\
  Computer Science and Software Engineering \\
  Butler University, Indianapolis IN, USA \\
  \texttt{sorenson@butler.edu} \\
  \texttt{http://www.butler.edu/$\sim$sorenson}
}

\date{\today}

\maketitle

\begin{abstract}
We present a randomized parallel algorithm that 
  computes the greatest common divisor
  of two integers of $n$ bits in length
  with probability $1-o(1)$
  that takes
  $O( n \log\log n / \log n)$ time using
  $O(n^{6+\epsilon})$ processors for any $\epsilon>0$
  on the EREW PRAM parallel model of computation.
The algorithm either gives a correct answer or reports failure.

We believe this to be the first randomized sublinear time 
  algorithm on the EREW PRAM for this problem.

\end{abstract}

\openup 0.5\baselineskip
  \section{Introduction}\label{sec:intro}

The parallel complexity of computing integer greatest common divisors
  is an open problem (see \cite{BGH82}),
  and no new complexity results have been published since the early 1990s.
This problem is not known to be either $\cal P$-complete or in $\cal NC$
  \cite{GHR,KR90,Reif}.

The first sublinear time parallel algorithm
  that uses a polynomial number of processors is due to
  Kannan, Miller, and Rudolph \cite{KMR87}.
Adleman and Kompella \cite{AK88} presented a randomized algorithm that runs in
  polylog time, but uses a superpolynomial, yet
  subexponential number of processors.
The fastest currently known algorithm is due to Chor and Goldreich \cite{CG90}
  which takes $O(n/\log n)$ time using $O(n^{1+\epsilon})$ processors.
See also \cite{Sorenson94}, and Sedjelmaci \cite{Sedjelmaci2008} who showed
  a clear way to do extended GCDs in the same complexity bounds.
However, all of these algorithms use the concurrent-read concurrent-write (CRCW)
  parallel RAM (PRAM) model of computation.

The algorithms of Chor and Goldriech \cite{CG90} 
  and the author \cite{Sorenson94}
  can be readily modified for the weaker 
  concurrent-read exclusive-write (CREW) PRAM
  to obtain running times of $O(n \log\log n / \log n)$ using a polynomial
  number of processors.
And of course one can take a CRCW PRAM algorithm and emulate it on an
  exclusive-read exclusive-write (EREW) PRAM at a cost of a factor of
  $O(\log n)$ in the running time, giving linear time algorithms for the
  EREW PRAM using a polynomial number of processors.

In this paper, we present what we believe is the first sublinear time,
  polynomial processor
  EREW PRAM algorithm for computing greatest common divisors.
Note that the EREW PRAM is weaker than the CREW or CRCW PRAM models of
  parallel computation.
We do make use of random numbers in a fundamental way.

\begin{thm}\label{thm:main}
  There exists a randomized algorithm to compute the 
  greatest common divisor of
  two integers of total length $n$ in binary 
  with probability $1-o(1)$ 
  in $O(n \log\log n / \log n)$ time
  using a polynomial number of processors on the EREW PRAM.
\end{thm}

In the next section we describe our algorithm, and in Section \ref{sec:anal}
  we prove correctness, give a complexity analysis, and flesh out the details
  of the algorithm.
We conclude in Section \ref{sec:smooth} 
  with a simple result on the relative
  density of integers with large polynomially smooth divisors, which is
  needed for the analysis of the algorithm.

\nocite{BS,CP}
  \section{Algorithm Description}\label{sec:alg}

Define the inputs as $u,v$ of total length $n$ in binary.
Let $B$, our small prime bound, be defined as
 $B = B(n):=n^2$.  
A larger value for $B$ can be chosen, so long as $\log B = o(n)$,
  but correctness would be compromised if $B$ were significantly smaller
  (see Section \ref{sec:correct}).
\begin{enumerate}
  \item
    Find a list of primes up to $B$.
    Also, for each prime $p\le B$, compute and save $p^e$ for
      $e=1\ldots \lfloor n/\log_2 p \rfloor$.

  \item
    Remove and save common prime factors of $u,v$ that are $\le B$,
      and let $u_0, v_0$ denote these modified inputs.
    WLOG we assume $u_0 \ge v_0$.

  \item
    \textbf{Main Loop}. Repeat while $u_iv_i\ne 0$.
      Here $i$ indicates the current loop iteration, starting at $i=0$.

    \begin{enumerate}
    \item
    For $j:=1$ to $2B\log n$ in parallel do:
    \begin{enumerate}
       \item 
           Choose $a_{ij}$ uniformly at random from $1\ldots v_i-1$.
       \item
           Compute $r_{ij}:= a_{ij}u_i \bmod v_i$.
       \item
         Compute $s_{ij}$ as $r_{ij}$ with all prime factors $p\le B$ removed.\\
         (We elaborate on how to do this below.)
    \end{enumerate}
    \item Find $s_i := \min_j \{ s_{ij} \}$.
       Let $j_{\min}$ denote the value of $j$ for which $s_i=s_{ij}$,
       and for later reference, let $a_i=a_{ij_{\min}}$.
    \item $u_{i+1}:=v_i$; $v_{i+1}:=s_i$.
    \end{enumerate}

  \item $u_i+v_i$ is, with probability $1-o(1)$,
    equal to $\gcd(u_0,v_0)$ (as we show below).
    If we err, it is by including spurious factors that do not belong,
      so verify that $u_i+v_i$ evenly divides both $u_0, v_0$, and if not,
      report an error.
    Otherwise, include any saved common prime factors found in step 2 above,
    and the algorithm is complete.

\end{enumerate}
  \section{Algorithm Analysis}\label{sec:anal}

In this section we
  prove correctness, and compute the parallel complexity of
  our algorithm from the previous section.

\subsection{Correctness}\label{sec:correct}

Note that in Steps 2 and 4 we handle any prime divisors $\le B$ 
  of the $\gcd(u,v)$, so WLOG we can assume either $\gcd(u,v)=1$ or
  $\gcd(u,v)>B$.

At iteration $i$ of the main loop, we perform the transformation
  $$ (u_i, v_i) \rightarrow (v_i, s_i). $$
Since $s_i$ is equal to $a_iu_i \bmod v_i$, ignoring factors below $B$,
  this transformation will only fail to preserve the greatest common divisor
  if $a_i$ and $v_i$ share a common factor.
Furthermore, this common factor must be composed only of primes
  exceeding $B$.
Since $a_i$ is chosen uniformly at random, the probability 
  $a_i$ and $v_i$ share a prime factor larger
  that $B$ is at most
$$
   \sum_{p|v_i,\, p>B} \frac{1}{p}
    \quad\le\quad 
   \sum_{p|v_i,\, p>B} \frac{1}{B} 
    \quad\le\quad 
    \frac{\log_B v_i}{ B} 
    \quad=\quad O\left( \frac{1}{n\log n} \right).
$$
As we will see below, with high probability, the number of main loop
  iterations is $o(n)$.
Thus, the probability that any of the $a_i$ values introduces a spurious
  factor is $o(1)$.

Note that in \cite{Sorenson03},
  a similar, but not identical, transformation was analyzed. 
It was observed that in practice,
  with \textit{no} removal of small prime divisors, the expected number of
  bits contributed by spurious factors was constant per main loop iteration.

\subsection{Runtime Analysis}

First we calculate the number of main loop iterations,
  and then we describe how each iteration can be computing in
  $O(\log n)$ time using a polynomial number of processors.

\subsubsection{Main Loop Iterations}

Let $W:= 0.5 (\log B)^2/\log\log B$.
Then by Theorem \ref{thm:smooth}, which we prove in the next section,
  the length of $s_{ij}$ is smaller than $r_{ij}<v_i$ by at least $\Theta(W)$ 
  bits with probability at least $1/B$.
(Note that we chose $0.5$ to get a clean $1/B$ probability - other choices
  for the constant can be made to work with the right adjustments.)

So, the probability all $2B\log n$ choices for $j$ fail to have
  $\log s_{ij} \le \log v_i - W$ is 
  $$\left( 1- \frac{1}{B} \right)^{2B\log n} 
     \quad=\quad O\left(\frac1{n^2}\right)  .$$
So, with probability $1- O(1/n^2)$, $\log s_i \le \log v_i - W$.

We remove roughly $(\log B)^2/\log \log B$ bits each main loop iteration.
Thus, the number of main loop iterations is
  $O( n \log\log B / (\log B)^2 ) = o(n)$.
The probability that \textit{any one} loop iteration fails to remove
  the needed $\Theta(W)$ bits is $O(1/n)$, so the probability we exceed this
  number of main loop iterations and terminate without
  computing an answer is $o(1)$.

\subsubsection{Computation Cost and Algorithm Details}

Unless stated otherwise, cost is given for the EREW PRAM.
For a brief overview of the cost of parallel arithmetic, see
  \cite[Section 6.2]{Sorenson94}.
\begin{description}
  \item[Step 1.]
    We can find the primes $\le B$ in $O(\log B)$ time using $O(B)$ processors
    (see \cite{SP94}).
    For each prime $p\le B$ and $e\le n$, we can compute $p^e$ in
     at most
      $O(\log n)$ multiplications, each of which takes $O(\log B)$ time
      using $B^{1+o(1)}$ processors \cite{SS71}.
      See also \cite[Theorem 12.2]{Reif}.
    As there are $O(B/\log B)$ primes, this is
    $O(\log n \log B)$ time using $nB^{2+o(1)}$ processors.
  \item[Step 2.]
    For each prime $p$ and exponent $e$, we assign a group of processors to
      see if $p^e$ divides $u$ but $p^{e+1}$ does not.
    Division takes $O(\log n)$ time using $n^{1+\epsilon}$ processors
      for any $\epsilon>0$ using
      the algorithm of Beame, Cook, and Hoover \cite{BCH86}, 
      giving a total processor count
      of $O(n^{2+\epsilon}B/\log B)$.

    The result is a vector of the form
      $[p_k^{e_k}]$ that lists the primes dividing $u$ with maximal exponents.
    Since there are at most $n/\log B$ integers in the vector $>1$,
      and they total at most $n$ bits (their product is $\le u$),
      the iterated product algorithm of \cite{BCH86} can take their product in
      $O(\log n)$ time using $n^{1+\epsilon}$ processors.
    Dividing $u$ by this product can be done at the same cost.

    We repeat this for $v$, and obtain a similar vector.

    We combine these two vectors using a minimum operation,
      and take the product of the entries, to obtain the
      shared prime power divisors of $u,v$ which must be saved for Step 4.

    The total cost of this step is $O(\log n)$ time using
      $O(n^{2+\epsilon}B/\log B)$ processors.

    See also \cite{DLX09} and references therein.
    
  \item[Step 3.]
     Checking for zero takes $O(\log n)$ time using $O(n)$ processors.

  \item[Step 3.(a)i]
      For each $j$, choosing an $n$-bit number at random 
      takes constant time using $O(n)$ processors.
      We reduce it modulo $v_i$ in $O(\log n)$ time using $n^{1+\epsilon}$
        processors \cite{BCH86}.
  \item[Step 3.(a)ii]
      This is simply a multiplication and a division, again taking
      $O(\log n)$ time using $n^{1+\epsilon}$ processors.
  \item[Step 3.(a)iii]
      Here we use the same method as described in Step 2 above.
      This is $O(\log n)$ time using
        $O(n^{2+\epsilon}B/\log B)$ processors for each $j$.
  \item[Step 3.(a)]
     And so, the total cost of this parallel step is
     $O(\log B)$ time using $O( n^{2+\epsilon}(\log n)B^2/\log B)$ processors.
  \item[Step 3.(b)]
     This can be done in $O(\log (B\log n))=O(\log B)$ 
       time using $O(B\log n)$ processors.
  \item[Step 3.(c)]
     This takes constant time using $O(n)$ processors.

     We conclude that the cost of one main loop iteration is
     $O(\log B)$ time using $O( n^{2+\epsilon}(\log n)B^2/\log B)$ processors.
     Step 3.(a)iii is the bottleneck.

     Earlier we showed that the number of iterations is
     $O(n \log\log B /(\log B)^2)$, for a total time of
     $O(n \log\log B /\log B)$ for all iterations of the the main loop.

  \item[Step 4.]
     This is an addition, a division, and a multiplication 
      using the results from Step 2;
      $O(\log n)$ time using $n^{1+\epsilon}$ processors.
\end{description}
Clearly, the bottleneck of the algorithm is Step 3.(a).
The overall complexity is
\begin{eqnarray*}
   O\left( \frac{ n \log\log B}{\log B} \right)
     &=&
   O\left( \frac{ n \log\log n}{\log n} \right) \quad
     \mbox{time, and} \\
   O\left( n^{2+\epsilon}(\log n)\frac{B^2}{\log B}\right)
     &=&
    O(n^{6+\epsilon}) \quad \mbox{processors, where $\epsilon>0$,}
\end{eqnarray*}
for the EREW PRAM.
This completes our proof of Theorem \ref{thm:main}.

\medskip

One could take $B$ to be superpolynomial in $n$;
  for example, if $B=\exp[\sqrt{n}]$ we can obtain a
  running time of roughly $\sqrt{n}$ using $\exp[ O(\sqrt{n}) ]$ processors.
Similar results could be obtained from some of the
  CRCW PRAM algorithms mentioned in the introduction by porting them to the
  EREW PRAM and setting parameters appropriately.

We can also obtain an $O(n/\log n)$ running time on the randomized CRCW PRAM;
  see \cite{BS07} for how to perform the necessary main loop operations
  in $O(\log n / \log\log n)$ time
  via the explicit Chinese remainder theorem.
See also \cite{DLX09}.

It would be interesting to see if this algorithm can be modified to
  compute Jacobi symbols quickly in parallel.
See \cite{MES98} and references therein.

  \section{Numbers with Smooth Divisors}\label{sec:smooth}

Let $P(n)$ denote the largest prime divisor of $n$.
If $P(n)\le y$ we say that $n$ is $y$-smooth.
Let
  $$ \Psi(x,y)= \#\{ n\le x \,:\,P(n)\le y \},$$
the number of integers $\le x$ that are $y$-smooth.
Let $u=u(x,y):=\log x/\log y$.
We will make use of the following lemma.
\begin{lemma}[{\cite[Corollary 1.3]{HT93}}]\label{lemma:smooth}
Let $\epsilon>0$ and assume $u<y^{1-\epsilon}$.  Then
$$ \Psi(x,y) = x u^{-u(1+o(1))} $$
for $x>y\ge 2$.
\end{lemma}
Note that the $o(1)$ here tends to zero for large $u$, and the
  implied constant depends on $\epsilon$.
Better results are known, but this suffices for our purposes.
For additional references see \cite{HT93,Tenenbaum},
  and for references on approximation algorithms for $\Psi(x,y)$
  see \cite{PS06}.

We recall the definition of $H_k$, the $k$th harmonic number as
$$ H_k = \sum_{i=1}^k \frac1i.$$
It is well known that
  $H_k = \log k + \gamma + O(1/k)$, where $\gamma=0.57721\ldots$ 
  is Euler's constant
(for example, see \cite[4.5.4]{PB}).

Fix a constant $c>0$.
Define $B(x)$ to be a strictly increasing function of $x$, but with
  $\log B(x) = o(\log x)$.
(We are primarily interested in $B(x)$ polynomial in $\log x$.)
Define
\begin{eqnarray*}
  W(x)&:=& \frac{c \cdot (\log B(x))^2}{\log\log B(x)}, \\
  F(x)&:=& \#\{ n\le x \,: \, n=my, \ P(m)\le B(x), \  \log m \ge W(x) \}.
\end{eqnarray*}
In other words,
  $F(x)$ counts integers $n\le x$
   where $n$ has a $B(x)$-smooth divisor that is $\ge\exp W(x)$.

\newcommand{\deltap}{(1+\delta)}
\newcommand{\deltam}{(1-\delta)}
\begin{thm}\label{thm:smooth}
Let $\epsilon>0$.
  For sufficiently
  large $x$ we have
  $$ F(x) \ge \frac{x}{ B(x)^{c(1+\epsilon)}}. $$
\end{thm}
\begin{proof}
Choose $\delta>0$ such that $\deltap^3 < 1+\epsilon$.
From the definition, we have
\begin{eqnarray*}
  F(x) &=& \sum_{y=1}^{x/\exp[W(x)]}\Psi\left(\frac{x}{y},B(x)\right).
\end{eqnarray*}
First, we limit the range of summation to obtain the lower bound
\begin{eqnarray*}
  F(x) & \ge &
  \sum_{y=x/(\exp[\deltap W(x)])}^{x/\exp[W(x)]}\Psi\left(\frac{x}{y},B(x)\right).
\end{eqnarray*}
Next, we apply Lemma \ref{lemma:smooth}.
We also observe that 
  $u^{-u(1+o(1))}\ge u^{-\deltap u}$ for large $u$,
  and for a lower bound, we can fix $u$ at its largest value on the
  interval of summation, namely $u= u(x)=\deltap W(x)/\log B(x)$.
This gives us
\begin{eqnarray*}
  F(x) & \ge &
  \sum_{y=x/(\exp[\deltap W(x)])}^{x/\exp[W(x)]}
     \frac{x}{y} \cdot  u^{-\deltap u}.
\end{eqnarray*}
Using $\sum_a^b 1/t = H_b-H_a \ge \deltam\log (b/a) $ 
  for sufficiently large $a$, we obtain that
\begin{eqnarray*}
  F(x) & \ge & x \cdot \delta\deltam W(x) \cdot u^{-\deltap u}  \\
       & \ge & x \cdot u^{-\deltap u} 
\end{eqnarray*}
as $W$ is a strictly increasing function of $x$ for large $x$.
Next we plug in for $u$ as follows:
\begin{eqnarray*}
  \log ( u^{-\deltap u} ) & = & -\deltap u \log u \\
    & = & -\deltap \frac{\deltap W(x)}{\log B(x)}
             \log \left(\frac{\deltap W(x)}{\log B(x)}\right) \\
    & = & -\deltap 
  \frac{\deltap\, c \, \log B(x)}{\log\log B(x)}
 \log \left( \frac{\deltap\, c \, \log B(x)}{\log\log B(x)} \right) \\
    &\ge & -  c\, \deltap^3 \log B(x)
\end{eqnarray*}
for $x$ sufficiently large.
We now have
\begin{eqnarray*}
  F(x) & \ge & x \cdot B(x)^{- c \deltap^3 } 
\end{eqnarray*}
for sufficiently large $x$.
\end{proof}

Only a lower bound is needed for our purposes, 
  but one can obtain an upper bound on $F(x)$ of similar shape
  using the same general methods.

\openup -0.5\baselineskip
\nocite{Sedjelmaci2001,Sedjelmaci2001b,Sedjelmaci2004,Sedjelmaci2008}


\end{document}